\documentclass[journal]{IEEEtran}

\ifCLASSINFOpdf
  \usepackage[pdftex]{graphicx}
\else
\fi

\usepackage{times,amsmath,color,amssymb,epsfig,cite,enumerate,setspace}
\usepackage{subfigure,multirow,bm,stfloats}
\usepackage[para]{threeparttable}

\usepackage{tabularx}
\usepackage{array}
\usepackage{booktabs}
\usepackage{graphicx}
\usepackage{graphics}
\usepackage{pbox}
\usepackage{dsfont}
\usepackage{psfrag}

\newtheorem{Remark}{\it Remark}[section]

\newtheorem{Proposition}{\it Proposition}[section]
\newtheorem{Lemma}{\it Lemma}[section]

\begin{document}

\title{Distributed Opportunistic Scheduling for Energy Harvesting Based Wireless Networks: A Two-Stage Probing Approach}

\author{Hang~Li, \IEEEmembership{Student Member,~IEEE,}
        Chuan~Huang, \IEEEmembership{Member,~IEEE,}
        Ping~Zhang, \IEEEmembership{Member,~IEEE,}
        Shuguang~Cui,~\IEEEmembership{Fellow,~IEEE,}
        and~Junshan~Zhang,~\IEEEmembership{Fellow,~IEEE}
\thanks{Part of this work appeared in the Proceedings of the IEEE International Conference on Computer Communications (INFOCOM), Toronto, ON, Canada, April 27 - May 2, 2014.}
\thanks{H.~Li and S.~Cui are with the Department of Electrical and Computer Engineering, Texas A\&M University, College Station, Texas, 77843 USA (e-mail: david\_lihang@tamu.edu; cui@ece.tamu.edu). S. Cui is also a Distinguished Adjunct Professor at King Abdulaziz University in Saudi Arabia and a Visiting Professor at ShanghaiTech University, China.}
\thanks{C.~Huang is with the National Key Laboratory of Science and Technology on Communications, University of Electronic Science and Technology of China, Chengdu, Sichuan 610051 China (e-mail: huangch@uestc.edu.cn).}
\thanks{P.~Zhang is with the School of Information and Communication and also with the State Key Lab. of Networking and Switching Technology, Beijing University of Posts and Telecommunications, Beijing, 100876 China (e-mail: pzhang@bupt.edu.cn).}
\thanks{J.~Zhang is with the School of Electrical, Computer and Energy Engineering, Arizona State University, Tempe, Arizona, 85287 USA (e-mail: junshan.zhang@asu.edu).}}

\maketitle

\begin{abstract}
  This paper considers a heterogeneous \emph{ad hoc} network with multiple transmitter-receiver pairs, in which all transmitters are capable of harvesting renewable energy from the environment and compete for one shared channel by random access. In particular, we focus on two different scenarios: the constant energy harvesting (EH) rate model where the EH rate remains constant within the time of interest and the i.i.d. EH rate model where the EH rates are independent and identically distributed across different contention slots. To quantify the roles of both the energy state information (ESI) and the channel state information (CSI), a distributed opportunistic scheduling (DOS) framework with two-stage probing and save-then-transmit energy utilization is proposed. Then, the optimal throughput and the optimal scheduling strategy are obtained via one-dimension search, i.e., an iterative algorithm consisting of the following two steps in each iteration: First, assuming that the stored energy level at each transmitter is stationary with a given distribution, the expected throughput maximization problem is formulated as an optimal stopping problem, whose solution is proven to exist and then derived for both models; second, for a fixed stopping rule, the energy level at each transmitter is shown to be stationary and an efficient iterative algorithm is proposed to compute its steady-state distribution. Finally, we validate our analysis by numerical results and quantify the throughput gain compared with the best-effort delivery scheme.
\end{abstract}

\section{Introduction}
Conventional wireless communication devices are usually powered by batteries that can provide stable energy supplies. However, the battery lifetime limits the operation time of such devices. Recently, energy harvesting (EH) techniques have been proposed as a promising alternative to the conventional constant power supplies \cite{SS,BM}, which is capable of transferring the renewable energy from the environment into electrical energy. In this way, the node lifetime can be prolonged significantly. Compared with the conventional constant energy suppliers, transmitters powered by energy harvesters are restricted by a new class of EH constraints, i.e., the consumed energy up to any time is bounded by the harvested energy until this point \cite{CKH}. Therefore, to meet certain performance requirements, such as throughput, stability, delay, etc., these EH constraints should be carefully taken into account in the design of EH-based communication systems.

\subsection{Related Works and Motivations}
Communication systems powered by energy harvesters have been investigated in recent years. For the point-to-point wireless systems, the authors in \cite{CKH} \cite{OO} considered the throughput maximization problem over a finite horizon for both the cases that the harvested energy information is non-causally and causally known to the transmitter, where the optimal solutions were obtained by the proposed one-dimension search algorithm and dynamic programming (DP) techniques, respectively. In \cite{HC}, the authors extended the results to the classic three-node Gaussian relay channel with EH source and relay nodes, where the optimal power allocation algorithms were proposed. With a more practical circuit model by considering the half-duplex constraint of the battery, the authors in \cite{ruizhang} proposed a save-then-transmit protocol, which divides each transmission frame into two parts: the first one for harvesting energy and the other for data transmission. For wireless networks with EH constraints, the authors in \cite{MACprotocols} investigated the performance of some standard medium access control protocols, e.g., TDMA, framed-Aloha, and dynamic-framed-Aloha.

In related works on \emph{ad hoc} networking, opportunistic scheduling has been known as an effective method to utilize the wireless resource \cite{mandrews,pv,Xinliu}. In particular, a distributed opportunistic scheduling (DOS) scheme was introduced in \cite{dos,Dongzheng}, where only local channel state information (CSI) is available to each transmitter. By applying optimal stopping theory \cite{optstop}, it has been shown in \cite{dos,Dongzheng} that the optimal solution for the expected throughput maximization problem has a threshold-based structure. When channel estimation is imperfect, the authors in \cite{twolevelchannel} proposed a two-level channel probing framework that allows the accessing transmitter to perform one more round of channel estimation before data transmission to improve the quality of estimated CSI and possibly increase the system throughput. The optimal scheduling policy of the two-level probing framework was proven to be threshold-based as well by referring to the optimal stopping with two-level incomplete information \cite{osptwo}.

Different from the traditional energy supplies (e.g., non-rechargeable batteries, power grid) in the conventional networks \cite{mandrews,pv,Xinliu, dos,Dongzheng,twolevelchannel}, we consider the network powered by energy harvesters that could generate electric energy from different renewable energy sources. Among various types of renewable energy sources, we consider two typical energy harvesting rate models in this paper\footnote{A more general case is that the transmitter only has causal information about EH rates, which could be modeled as a Markov process. This model has been used in the point-to-point wireless system \cite{CKH,OO}.}:
\begin{enumerate}
  \item \emph{Constant energy harvesting rate model}: The EH rate (specifically, the amount of harvested energy per unit time) can be approximated as a constant within the entire time duration of interest. For example, the power variation coherence time of wind and solar EH systems is on the order of multiple seconds \cite{MB,energyreport}, while the duration of one communication block is about several milliseconds. Thus, over thousands of communication blocks, the EH rate keeps almost the same.
  \item \emph{Independent and identically distributed (i.i.d.) energy harvesting rate model}: Compared to the constant rate model, the EH rate for this case changes much faster, i.e., comparable to the duration of one communication block. For example, the energy from light, thermal, kinetic, or ambient-radiation sources, usually changes every several milliseconds. Accordingly, EH rates can be modeled as an i.i.d. \cite{OO,MACprotocols} random process.
\end{enumerate}

With the above two EH models, we investigate the DOS problem for a heterogeneous EH-based network, where the channel gains across different links and the EH rates across different transmitters are non-identical. The system works in a two-stage pattern as follows. In the first stage, all transmitters adopt random access and do channel probing (CP), during which the successful link can obtain the CSI via channel contentions, similar to those in \cite{dos,Dongzheng,twolevelchannel}. In the second stage, the successful transmitter at the first stage has the option to spend certain time to harvest more energy, i.e., executes energy probing (EP); and then, with the updated energy state information (ESI), it decides either to transmit in the rest of the transmission block, or to stop probing and give up the channel. With EP, \emph{since the total duration of the transmission block is fixed, although spending more time on harvesting energy could increase the energy level, it decreases the portion of the time for data transmission, which leads to a tradeoff to optimize}.

\subsection{Summary of Contributions}
We propose a DOS framework for an \emph{ad hoc} network powered by energy harvesters, which efficiently utilizes both the CSI and the ESI at each transmitter. In this framework, we adopt a ``save-then-transmit'' scheme, i.e., the transmitter keeps harvesting energy before it initiates the transmission that uses up all the available energy in the battery. Note that such a greedy power utilization scheme is suboptimal in general, while it is sensible when the number of transmitters is large.

The main contributions of this paper are summarized as follows:
\begin{enumerate}
  \item First, by assuming that the battery state at each transmitter is stationary with a certain distribution, the throughput maximization problem for the considered network is cast as a rate-of-return problem. We prove the existence of the optimal stopping rules for both EP and CP, and further obtain:
      \begin{itemize}
        \item For the constant EH model, the optimal stopping rule of EP is determined by maximizing the throughput over the transmission block before starting EP, and it is either zero or a finite value according to the given CSI and ESI. Then, based on the stopping rule of EP, the optimal stopping rule of CP is shown to be a pure threshold policy (the threshold does not change over time) and the transmission decision is made right after each round of CP.
        \item For the i.i.d. EH model, the optimal stopping rule for EP is shown to be dynamic and threshold based, which is obtained by solving a stopping problem over a finite-time horizon. The stopping rule of CP is also threshold based and obtained based on the decision of EP, i.e., either transmit or start a new CP. Unlike the constant case, the transmission decision under i.i.d. EH model is made during the process of EP.
      \end{itemize}
  \item Next, with a fixed stopping rule, we show the existence of the steady-state distribution of the battery state by constructing a ``super'' Markov chain with its states being jointly determined by all transmitters. Moreover, we propose an efficient iterative algorithm to compute the steady-state distribution, executed at each transmitter in parallel. Particularly, it is shown that with the constant EH model, if the network consists of $n$ transmitters and each one is with $m$ possible energy states, the computational complexity for one iteration of the proposed algorithm is on the order of $O\left(n^2m^2\right)$, which is more efficient (when $n$ and $m$ are large) than that of the super Markov chain case, whose complexity for one iteration is on the order of $O\left(2m^{2n}\right)$.
  \item Finally, by exploiting the structure of the rate-of-return problem, we show that the maximum throughput and the optimal scheduling strategy of the DOS framework could be obtained for both the two EH rate models, via one-dimension search by repeating the above two steps.
\end{enumerate}

The rest of this paper is organized as follows. Section \ref{systemmodel} introduces the system model. In Section \ref{formulation}, the throughput maximization problem is formulated and solved under the assumption that the stationary distribution of the battery at each transmitter is known. Then, with the obtained stopping rule, we prove in Section \ref{markov} the existence of the steady-state distribution for each transmitter, and propose an iterative algorithm to compute it. Section \ref{onedsearch} discusses the computation for the optimal throughput. In Section \ref{numerical}, numerical results are provided to validate our analysis and evaluate the throughput gain of our proposed scheduling scheme against the best-effort delivery. Finally, Section \ref{fin} concludes the paper.

\section{System Model}\label{systemmodel}
We consider a heterogeneous single-hop \emph{ad hoc} network, where all the $I$ transmitter-receiver pairs have independent but not necessarily identical statistical information of CSI and ESI. All pairs contend for one shared channel by random access. For each link, the transmitter is powered by a renewable energy source and utilizes a small rechargeable battery to temporally store the harvested energy. Note that the transmitter could keep harvesting energy until it initiates a data transmission. In addition, we do not consider the effect of inefficiency in energy storage and retrieval, nor the energy consumed other than data transmission, which can be approximately neglected by properly adjusting the energy model \cite{CKH,HC,OO,MACprotocols}. Denote the duration of one channel contention as $l>0$, and the length of one transmission block as $L$, which is an integer multiple of $l$.

\begin{figure}
  \centering
  \includegraphics[width=3.3in]{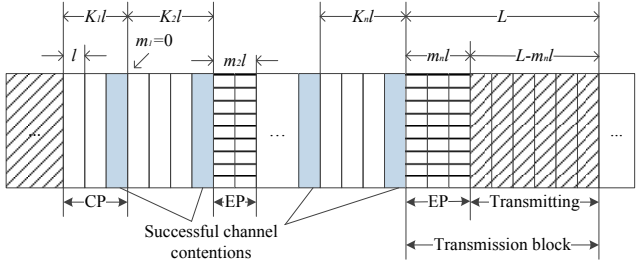}
  \caption{One realization for the DOS with two-stage probing.}
  \label{onerealization}
\end{figure}
As illustrated in Fig.~\ref{onerealization}, the DOS procedure of the whole network takes place in two stages: First, each transmitter probes the channel via random access and harvests energy at the same time; and then the successful transmitter may start the EP (to potentially increase the average transmission rate over the transmission block\footnote{If the successful transmitter experiences a bad channel condition and a low energy level, it may skip the transmission.}) before the data transmission process.

\subsubsection{Channel probing}
In the first stage, a successful channel contention is defined as follows: All transmitters first independently contend for the channel until there is only one contending in a particular time slot. Furthermore, one round of CP is defined as the process to achieve one successful channel contention. Denote the probability that transmitter $i$ contends for the channel as $q_{i}$, $1\leq i\leq I$, with $0\leq q_{i}\leq 1$. As such, the probability that the $i$-th transmitter successfully occupies the channel is given by $Q_i=q_{i}\prod_{j\neq i}(1-q_{j})$. Then, the probability to achieve one successful channel contention at each time slot is given by $Q=\sum_{i=1}^{I}Q_i$, and it is easy to check that $Q\leq 1$ \cite{thrpbc}. Accordingly, for the $n$-th round of CP, $n\geq1$, we use $K_n$ to denote the number of time slots needed to achieve a successful channel contention, which is a random variable and satisfies the geometric distribution with parameter $Q$ \cite{dos,Dongzheng,twolevelchannel}. In this way, the expected duration of one round of CP is given as $l/Q$. Denote the transmitted signal at transmitter $i$ as $x^i$, and the received signal $y^i$ is thus given by $ y^i=h^ix^i+z^i$, where $h^i$ is the complex channel gain and $z^i$ is the circularly symmetric complex Gaussian (CSCG) noise with zero mean and variance $\sigma^2$ at the receiver. Across different links, $\{h^i\}_{1\leq i\leq I}$ are independent with finite mean and variance, while not necessarily identically distributed. After one round of CP, the successful transmitter can perfectly estimate the corresponding channel gain via certain feedback mechanisms, and thus $h^i$ is assumed a known constant during the whole transmission block. After CP, the successful transmitter chooses one of the following actions based on its local CSI and ESI:

(a) releases the channel (if the CSI and ESI indicate that the transmission rate is lower than a threshold) and let all links re-contend; or

(b) directly transmits until the end of the transmission block; or

(c) holds the channel, starts EP.\\
Note that to complete one data transmission, it may take $n$ rounds of CPs as depicted in Fig.~\ref{onerealization}. It is worth noting that each transmitter keeps harvesting energy until it starts a transmission, and after each round of CP, only the successful transmitter makes a choice among three actions as listed above.

\subsubsection{Energy Probing}
When the successful transmitter decides not to take action (a) or (b) defined above, it starts the second stage EP, i.e., action (c), to obtain more energy. During this stage, the transmitter chooses to continue harvesting energy slot by slot, and then ends EP by action (a) or (b), i.e., either releasing the channel or transmitting over the rest of the transmission block. As it is depicted in Fig.~\ref{onerealization}, one transmission is fulfilled with $n$ rounds of CPs and $m_n$ extra slots of EP.

For transmitter $i$, let $B_{n,m}^i\in\mathbf{\Delta}$ denote the energy level of the battery after the $n$-th round of CP and $m$ additional time slots for EP, where $\mathbf{\Delta}=\left\{0, \delta, 2\delta,\cdots, B_{max}\delta\right\}$ is the set of all possible energy states, with $\delta$ being the minimum energy unit and $B_{max}\delta$ the capacity of the battery. We use $E^i_t$ to denote the EH rate of transmitter $i$ at time $t$. As noted in the previous section, we consider the following two types of scenarios:
\begin{enumerate}
  \item \emph{Constant EH rate model}: $\left\{E_{t}^i\right\}_{t\geq1}$ are constants for each $i$, i.e., $E_{t}^i=E^i\in\mathbf{\Delta}$ for all $t\geq1$, and $\{E^i\}$ can thus be learned and assumed non-causally known before transmissions.
  \item \emph{I.i.d EH rate model}: The EH rates among different transmitters are independent. For transmitter $i$, $\left\{E_{t}^i\right\}_{t\geq1}$ are i.i.d. across $t$, with finite mean and the probability mass function (PMF) $\Pr\{E_{t}^i=e\delta\}=F^i(e)$, where $e\in\left\{0,1,2,\cdots\right\}$.
\end{enumerate}
Under the save-then-transmit scheme, the energy level will keep non-decreasing and drop to zero after the transmission, which forms a Markov chain (as described in Section \ref{markov} later). Thus, the energy level $B_{n,m}^i$ can be written as
\begin{equation}\label{energylevel}
   B_{n,m}^i=\min\left\{B_{n,0}^i+l\sum_{k=0}^m E^i_k,B_{max}\delta\right\},
\end{equation}
where $n\geq 1$, $0\leq m\leq L/l$, and $\min\{x,y\}$ denotes the smaller value between two real numbers $x$ and $y$. Note that $B_{n,0}^i$ indicates the energy level after the successful contention round before taking any action. If $m=0$, i.e., transmitter $i$ does not do EP, we let $\sum_{k=0}^m E^i_k=E_0^i=0$.

\section{Transmission Scheduling}\label{formulation}
In this section, we target to derive the optimal scheduling policy that maximizes the average throughput for the considered network with the proposed two-stage access strategy, conditioned on the given battery state distribution. We point out that the results obtained in this section are based on the assumption that the energy level at transmitter $i$ is stationary with a given distribution $\Pi^i$, for $1\le i\le I$, which will be validated in Section \ref{markov}.

\subsection{Problem Formulation}\label{subformulation}
After the $n$-th round of CP and $m$ additional time slots, the CSI and the ESI at the successful transmitter are given as $\mathcal{F}_{n,m}^i=\left\{h_n^i,B_{n,m}^i\right\}$. Note that the channel gain $h_n^i$ is now indexed by $n$, which is determined at the end of the $n$-th round of CP and assumed fixed during the whole data transmission block. In particular, $\mathcal{F}_{n,0}^i=\left\{h_n^i,B_{n,0}^i\right\}$ denotes the initial information right after the $n$-th round of CP. For convenience, we omit the index $i$ for either the CSI or the ESI in the sequel, and retrieve it when necessary.

By adopting the save-then-transmit scheme at the transmitters to fully take advantage of each channel use, the transmission rate over $L/l$ time slots with state $\mathcal{F}_{n,m}$ is defined as
\begin{align}\label{averagerate}
    R_{n}(m)=\left(1-\frac{ml}{L}\right)\log\left(1+|h_n|^2\frac{B_{n,m}}{(L-ml)\sigma^2}\right).
\end{align}
When $ml=L$, we set $R_{n}(m)=0$ since there is no transmission in this case.
\begin{Remark}
  Some important properties of $R_{n}(m)$ are listed as follows.
  \begin{itemize}
    \item $\mathbb{E}\left[R_{n}(m)\right]<\infty$ and $\mathbb{E}\left[(R_{n}(m))^2\right]<\infty$, which results from the fact that $h_n$ has finite mean and variance and the energy level $B_{n,m}$ is also finite.
    \item $\{R_{n}(m)\}_{n\geq1}$ are approximately independent random variables over $n$. To see this, recall that the channel gains and the battery states are independent across different transmitters at a given time slot; moreover, the probability is small for a transmitter to occupy the channel in two consecutive contentions when the number of user pairs is large. For example, in an \emph{ad hoc} network with $K$ pairs where each pair fairly competes for the channel use with probability $1/K$, such a probability is $\frac{1}{K^2}(1-1/K)^{2(K-1)}$ \cite{thrpbc}, which is as small as 0.0625 even when $K=2$. Thus,  $\{\mathcal{F}_{n,m}\}_{n\ge1}$ are nearly independent over $n$, which implies that $\{R_{n}(m)\}_{n\geq1}$ are independent over $n$.
  \end{itemize}
\end{Remark}

Let $N$ be the stopping rule for CP, and $M_n$ be the stopping rule for EP associated with the $n$-th CP for $1\le n\le N$, which together tell the transmitter when to start the data transmission. Then, under these stopping rules, the transmission rate would be $R_{N}(M_N)$, and we let $T_{N}$ be the total time duration for completing one data transmission. Here, $T_N$ contains the duration of $N-1$ rounds of CP, which is given by $l\sum_{n=1}^{N-1}K_n$, and $l\sum_{n=1}^{N-1}M_n$ time slots in which the transmitter probes the energy but gives up the channel after EP. Also, after the $N$-th round of CP with the time $K_Nl$, the transmitter may use $M_N$ slots for the EP and transmit within the duration $L-M_Nl$ afterwards. Accordingly, we obtain
\begin{align}\label{tnm}
  T_{N}=l\sum_{n=1}^{N-1}M_n+l\sum_{n=1}^{N}K_n+L.
\end{align}
If such a process is executed $J$ times with $R_{N_j}(M_{N_j})L$ bits transmitted at each transmission, $1\leq j\leq J$, we obtain the average throughput $\lambda$ per transmission of the network:
\begin{equation}
\frac{L\sum_{j=1}^JR_{N_j}(M_{N_j})}{\sum_{j=1}^JT_{N_j}} \longrightarrow \lambda=\frac{L\mathbb{E}\left[R_{N}(M_N)\right]}
             {\mathbb{E}\left[T_{N}\right]}~\hbox{a.s.}\nonumber
\end{equation}
as $J\rightarrow\infty$ by the renewal theory \cite{PB}. Again, we point out that the energy level is stationary at the $N_j$-th round of CP for $j\geq1$, as we assumed.

Our target is to maximize $\lambda$ by adjusting the stopping rule $N$ and $\{M_n\}_{1\le n\le N}$. It is easy to see that maximizing $\lambda$ is in fact a ``rate-of-return'' stopping problem \cite{optstop,somepro} (for which the specific definition is given later). Instead of directly solving this problem, we examine the ``net reward'' of the considered network, which is given as
\begin{align}
&r_{N}(\lambda)=R_{N}(M_N)L-\lambda T_N\nonumber\\
=&(R_{N}(M_N)-\lambda)L-\lambda l\left[K_N+\sum_{n=1}^{N-1}\left(K_n+M_n\right)\right],\label{rnlambda}
\end{align}
for some $\lambda>0$. The term $(R_{N}(M_N)-\lambda)L$ can be interpreted as the reward of transmission, $\lambda l K_n$ as the cost of CP, and $\lambda lM_n$ as the cost of failed EP for $1\leq n\leq N-1$. We set $r_{-\infty}(\lambda)=-\infty$ since it is irrational that the system does not send any data forever. Then, we define the maximum value of the expected net reward with $\lambda>0$ as
\begin{equation}\label{netreward}
    S^*(\lambda)=\sup_{N\in \mathcal{N},\{M_n\}_{1\le n\le N}}\mathbb{E}\left[r_{N}(\lambda)\right],
\end{equation}
where $\sup(\cdot)$ denotes the least upper bound for a set of real numbers, and
\begin{align}
   \mathcal{N}\triangleq&\left\{N: N\geq1,~ \mathbb{E}\left[T_{N}\right]<\infty,\right.\nonumber\\
   &\left.~\hbox{for}~M_n\in[0,L/l]~\hbox{with}~1\leq n\leq N\right\}.
\end{align}
\begin{Remark}
  One important property of problem (\ref{netreward}) is time invariance. We observe that before the system starts the $N$-th round of CP, the accumulated cost $\lambda l\sum_{n=1}^{N-1}\left(K_n+M_n\right)$ over the past $N-1$ rounds of CP has already been finalized, with no need to be further considered in the remaining decision process. Moreover, $\{R_{n}(M_n)\}_{1\le n\le N}$ are independent over $n$ as we mentioned before; it follows that the expected optimal reward before the $N$-th round of CP is the same as that of any previous round of CP. In other words, the system can obtain the expected optimal reward $S^*(\lambda)$ whenever a new round of CP is about to start. Therefore, we conclude that problem (\ref{netreward}) is \emph{time invariant}.
\end{Remark}

Recall from Section \ref{systemmodel} that after each round of CP, the successful transmitter will choose one of three actions (i.e., transmitting, giving up the channel, or starting EP) according to the stopping rule of CP, which needs the expected reward of EP depending on the stopping rule of EP. Thus, we will first introduce the formulation and the optimal stopping rule for EP, and then for CP.

\subsubsection{Formulation for EP}
When the successful transmitter starts EP after the $n$-th round of CP, where $1\le n\le N$, it will end up with one of the two actions: transmitting or giving up the channel without transmission. Specifically, we define the expected optimal reward at the $k$-th slot of EP, $0\leq k\leq L/l$, as
\begin{align}\label{uk}
U_k(\mathcal{F}_{n,k})=\max_{k\leq M_n\leq L/l}\mathbb{E}&\left[\max\left\{(R_{n}(M_n)-\lambda)L,\right.\right.\nonumber\\
    &\left.\left.-\lambda lM_n+S^*(\lambda)\right\}\mid\mathcal{F}_{n,k}\right],
\end{align}
where $-\lambda lM_n+S^*(\lambda)$ is the expected value of giving up the channel after $M_n$ slots of EP. If $k=0$, $U_0(\mathcal{F}_{n,0})$ denotes the maximum of the expected net reward right after the $n$-th round of CP. In other words, we want to find the optimal stopping rule $M^*_n$ of EP which attains
\begin{align}\label{u0}
U_0(\mathcal{F}_{n,0})=\max_{0\leq M_n\leq L/l}\mathbb{E}&\left[\max\left\{(R_{n}(M_n)-\lambda)L,\right.\right.\nonumber\\
    &\left.\left.-\lambda lM_n+S^*(\lambda)\right\}\mid\mathcal{F}_{n,0}\right].
\end{align}
Note that $M^*_n$ exists since problem (\ref{u0}) is an optimal stopping problem over a finite time horizon \cite{optstop,goran}.

\subsubsection{Formulation for CP}
By choosing $\{M_n^*\}_{1\le n\le N}$, we define
\begin{equation}\label{optirate}
\lambda^*=\sup_{N\in \mathcal{N}}\frac{L\mathbb{E}\left[R_{N}(M_N^*)\right]}
             {\mathbb{E}\left[T_{N}\right]},~N^*=\arg\sup_{N\in \mathcal{N}}\frac{L\mathbb{E}\left[R_{N}(M_N^*)\right]}
             {\mathbb{E}\left[T_{N}\right]}.
\end{equation}
Note that if the optimal stopping rule $N^*\notin\mathcal{N}$, we would claim that $N^*$ does not exist. Thus, $\lambda^*$ is the optimal average throughput of the original rate-of-return problem.

The connection between the transformed problem (\ref{netreward}) and the original problem (\ref{optirate}) is introduced in the following lemma. It is worth noticing that with the optimal stopping rule $\{M_n^*\}_{1\le n\le N}$ for EP, problem (\ref{netreward}) boils down to a one-level stopping problem with stopping rule $N$.
\begin{Lemma}\label{lemforcontwo}

(i) If there exists $\lambda^*$ such that $S^*(\lambda^*)=0$, this $\lambda^*$ is the optimal throughput defined in (\ref{optirate}). Moreover, if $S^*(\lambda^*)=0$ is attained at $N^*(\lambda^*)$, the stopping rule $N^*$ defined in (\ref{optirate}) is the same as $N^*(\lambda^*)$, i.e., $N^*=N^*(\lambda^*)$.

(ii) Conversely, if (\ref{optirate}) is true, there is $S^*(\lambda^*)=0$, which is attained at $N^*$ given by (\ref{optirate}).
\end{Lemma}
This lemma directly follows Theorem 1 in Chapter 6 of \cite{optstop}.


The next proposition secures the existence of the optimal stopping rule for CP.
\begin{Proposition}\label{lemma1}
With the EP stopping rule $\{M_n^*\}_{0\le n\le N}$, the optimal stopping rule $N^*(\lambda)$ for problem (\ref{netreward}) exists. Moreover, for $N\geq1$, the following equation holds
\begin{equation}\label{optiequi}
S^*(\lambda)=U_0(\mathcal{F}_{N,0})-\lambda lK_N.
\end{equation}
\end{Proposition}
The proof is given in Appendix A.
\begin{Remark}
The equation (\ref{optiequi}) is obtained from the {\it optimality equation} of the CP. The calculation of the optimal throughput relies on this equation, which will be shown in Section \ref{onedsearch}.
\end{Remark}

Now, we are ready to derive the optimal stopping rules $N^*$ and $\{M_n^*\}$ that jointly maximize the expected value of $r_{N}(\lambda)$ for the two different EH models. As we mentioned above, the stopping rule $N$ for CP relies on the form of $M_N$ (the stopping rule for EP). We will find the optimal stopping rule $M_N^*$ before $N^*$. After obtaining the forms of the optimal stopping rules, the calculation for the optimal throughput will be discussed.

\subsection{Optimal Stopping Rule for Constant EH Model}\label{const}
For notation simplicity, we omit the index $N$ of CP when we derive the stopping rule $M$ in this subsection. Then, we will derive the stopping rule $N$ based on the results of EP.

When the EH rate is constant, the transmission rate $R(M)$ is deterministic for a given $\mathcal{F}_{0}$ over the transmission block. Then, we obtain a simplified version of $U_0(\mathcal{F}_{0})$ (\ref{u0}) as
\begin{align}
 U_0(\mathcal{F}_{0})=\max_{0\leq M\leq L/l}\max\left\{(R(M)-\lambda)L,-\lambda lM+S^*(\lambda)\right\}.\nonumber
\end{align}
The value of $U_0(\mathcal{F}_{0})$ can be obtained simply by comparing $-\lambda lM+S^*(\lambda)$ and $(R(M)-\lambda)L$, whose values can be computed individually. Clearly, the first one achieves its maximum $S^*(\lambda)$ at $M=0$. For the second term, only $R(M)$ is changing over $M$ with a given $\mathcal{F}_{0}$. Therefore, we settle down to the following auxiliary problem:
\begin{equation}\label{vn}
  V^*=\arg\max_{0\leq V\leq L/l}R(V).
\end{equation}
Then, we could use the optimal $V^*$ to find $M^*$ without difficulty. Note that when $Vl=L$, it follows that $R(V)=0$ according to our definition in Section \ref{systemmodel}, which implies that $V=L/l$ cannot be optimal, and thus we take $0\leq V\leq L/l-1$. We first consider a related continuous version of $R(V)$ by relaxing $Vl/L$ as $\rho$, $0\leq\rho<1$:
\begin{align}\label{optione}
  &\max_{0\leq\rho<1}R(\rho)=\max_{0\leq\rho<1}(1-\rho)\nonumber\\
&~~~\cdot\log\left(1+|h|^2\frac{\min\{B_{0}+\rho LE,B_{max}\delta\}}{(1-\rho)L\sigma^2}\right).
\end{align}
After solving (\ref{optione}), we will show how to obtain the optimal solution of problem (\ref{vn}).

First, we establish some properties for the objective function of problem (\ref{optione}).
\begin{Proposition}\label{lemma2}
For arbitrary $a,b\geq0$, we have that
\begin{enumerate}
      \item the function $y(x)=(1-x)\log\left(1+\frac{a+bx}{1-x}\right)$ is concave over $[0,1)$, and $\lim_{x\rightarrow1^-}y'(x)<0$;
      \item the function $g(x)=(1-x)\log\left(1+\frac{a}{1-x}\right)$ is concave and non-increasing over $[0,1)$.
\end{enumerate}
\end{Proposition}
\begin{proof}
Please see Appendix B.
\end{proof}

Since $\rho\in[0,1)$, when $\frac{B_{max}\delta-B_{0}}{LE}\geq1$, $R(\rho)$ is simply concave over $\rho$ on $[0,1)$ according to part 1) of Proposition \ref{lemma2}. When $\frac{B_{max}\delta-B_{0}}{LE}<1$, according to Proposition \ref{lemma2}, $R_{N}(\rho)$ is concave over $\left[0,\frac{B_{max}\delta-B_{0}}{LE}\right]$, and is non-increasing on $\left[\frac{B_{max}\delta-B_{0}}{LE},1\right)$. Thus, $R(\rho)$ cannot achieve its maximum on $\left(\frac{B_{max}\delta-B_{0}}{LE},1\right)$. Therefore, we treat this fact as a new constraint over $\rho$, and rewrite problem (\ref{optione}) as
\begin{align}\label{optig}
  &\max G(\rho)=\max(1-\rho)\log\left(1+|h|^2\frac{B_{0}+\rho LE}{(1-\rho)L\sigma^2}\right)\nonumber\\
  &~\hbox{s.t.} ~~B_{0}+\rho L E\leq B_{max}\delta,~ 0\leq\rho<1.
\end{align}

Next, we establish the following proposition to solve problem (\ref{optig}), where the obtained solution is optimal for problem (\ref{optione}) as well.
\begin{Proposition}\label{prop1}
The optimal solution $\rho^*$ for problem (\ref{optig}) is given by:
\begin{equation}
  \rho^*=\left\{
  \begin{array}{ll}
\min\left\{\rho_0,\frac{B_{max}\delta-B_{0}}{LE}\right\}, & \hbox{when $\frac{C+D}{1+C}\geq\log(1+C)$;} \\
0, & \hbox{otherwise,}
  \end{array}
\right. \nonumber
\end{equation}
where $C=\frac{|h|^2B_{0}}{L\sigma^2}$, $D=\frac{|h|^2E}{\sigma^2}$, and $\rho_0$ is the unique solution for the equation $\log\left(1+\frac{C+D\rho}{1-\rho}\right)=\frac{C+D}{1-\rho+C+D\rho}$ when $\frac{C+D}{1+C}\geq\log(1+C)$.
\end{Proposition}
\begin{proof}
Please see Appendix C.
\end{proof}

Based on the optimal solution $\rho^*$, the optimal $V^*$ for $R(V)$ in (\ref{vn}) can be obtained easily: We only need to compare $R(\left\lfloor\rho^*L/l\right\rfloor)$ against $R(\left\lceil\rho^*L/l\right\rceil)$, and $V^*$ should attain the larger value. Specifically, we have the following result.
\begin{Proposition}\label{prop2}
The optimal $V^*$ of the problem (\ref{vn}) is given by
\begin{equation}\label{stoppingM}
    V^*=\left\{
             \begin{array}{ll}
               \left\lfloor\rho^*L/l\right\rfloor,  &\hbox{if $R(\left\lfloor\rho^*L/l\right\rfloor)\geq R(\left\lceil\rho^*L/l\right\rceil)$;} \\
                \left\lceil\rho^*L/l\right\rceil, &\hbox{if $R(\left\lceil\rho^*L/l\right\rceil)> R(\left\lfloor\rho^*L/l\right\rfloor)$;}\\
              0, &\hbox{otherwise.}
             \end{array}
           \right.
\end{equation}
where $\rho^*$ is obtained by Proposition \ref{prop1}. Thus, the optimal stopping rule $M^*$ is given by
\begin{equation}
    M^*=\left\{
             \begin{array}{ll}
              0,  &\hbox{if $(R(V^*)-\lambda)L<S^*(\lambda)$;} \\
              V^*, &\hbox{otherwise.}
             \end{array}
           \right.
\end{equation}
The optimal reward $U_0(\mathcal{F}_{0})$ with constant EH rate model is
\begin{align}\label{u0withm}
 U_0(\mathcal{F}_{0})=\max\left\{(R(V^*)-\lambda)L,S^*(\lambda)\right\}.
\end{align}
\end{Proposition}

Next, the following proposition formally quantifies the optimal stopping rule $N^*$ and the equation to compute the optimal throughput $\lambda^*$.
\begin{Proposition}\label{stopingruleN}
The optimal stopping rule to solve problem (\ref{netreward}) is given by
\begin{equation}\label{cpruleconst}
   N^*=\min\left\{n\geq1: R_{n}(V^*)\geq\lambda^*\right\},
\end{equation}
with $V^*$ given in Proposition \ref{prop2}. Moreover, $\lambda^*$ satisfies the following equation
\begin{equation}\label{lambdaequi}
    \sum_{i=1}^I Q_i\mathbb{E}\left[\left(R^i\left(V^{*}\right)-\lambda^*\right)^+\right]=\frac{\lambda^*l}{L},
\end{equation}
where the function $(x)^+$ means $\max\{x,0\}$ for some real number $x$, and $Q_i$ is the probability of a successful channel contention at transmitter $i$, defined in Section \ref{systemmodel}. The index $n$ for $R^i\left(V^{*}\right)$ in (\ref{lambdaequi}) is removed since $\{R_n\left(V^{*}\right)\}_{n\geq1}$ are ergodic for $1\le i\le I$.
\end{Proposition}
\begin{proof}
Following (\ref{u0withm}) in Proposition \ref{prop2}, the stopping rule $N^*$ has the form
\begin{equation}
   N^*=\min\left\{n\geq1: (R_{n}(V^*)-\lambda^*)L\geq S^*(\lambda^*)\right\}.\label{cpruleconst2}
\end{equation}
Thus, we can obtain $N^*$ by plugging $S^*(\lambda^*)=0$ into (\ref{cpruleconst2}), which results in (\ref{cpruleconst}). Finally, equation (\ref{lambdaequi}) can be obtained by plugging $S^*(\lambda^*)=0$ into (\ref{optiequi}) and taking the expectation on both sides.
\end{proof}
\begin{Remark}
Note that the stopping rule (\ref{cpruleconst2}) implies that each transmitter has the same threshold that is globally determined even when all transmitters have different statistics of the CSI and ESI. The intuition is similar to that in \cite{Dongzheng}: In order to guarantee the overall system performance, the transmitter with a bad channel condition and a low energy level should ``sacrifice'' its own reward, while the one with good conditions should transmit more data.
\end{Remark}

Directly following Propositions \ref{prop2} and \ref{stopingruleN}, the next proposition gives the DOS under the constant EH model.
\begin{Proposition}\label{theo1}
After the $n$-th round of CP, it is optimal for the successful transmitter to take one of the following two options:
\begin{enumerate}
  \item release the channel immediately if $R_{n}(V^*)<\lambda^*$ (which is equivalent to $M^*=0$), and let all transmitters perform the next round of CP;
  \item otherwise, transmit after $V^*$ slots for EH, where $V^*$ is given by Proposition \ref{prop2}.
\end{enumerate}
\end{Proposition}

\subsection{Optimal Stopping Rule for i.i.d. EH Model}\label{rand}
Similarly as in the previous subsection, we first consider problem (\ref{u0}) to find the optimal stopping rule $M^*$, then the optimal stopping rule $N^*$ afterwards.

Under the i.i.d. EH model, $U_0(\mathcal{F}_{0})$ has the form in (\ref{u0}). As we mentioned in Section \ref{subformulation}, it is a finite-horizon stopping problem \cite{optstop,goran}, and the solution of problem (\ref{u0}) could be directly generalized in the next proposition.
\begin{Proposition}\label{propEPrand}
For $0\leq k\leq L/l$ and some $\lambda>0$, the \emph{optimality equation} for problem (\ref{u0}) is given by
\begin{align}\label{optiequirandEP}
U_k(\mathcal{F}_{k})=\max&\left\{(R(k)-\lambda)L,-\lambda kl+S^*(\lambda),\right.\nonumber\\
&\left.\mathbb{E}[U_{k+1}(\mathcal{F}_{k+1})\mid\mathcal{F}_{k}]\right\},
\end{align}
and the optimal stopping rule has the following form:
\begin{align}\label{optiruleEPrand}
M^*&=\min\left\{0\leq k\leq L/l:\right.\nonumber\\
&~\left.U_k(\mathcal{F}_{k})=\max\{(R(k)-\lambda)L,-\lambda kl+S^*(\lambda)\}\right\}.
\end{align}
\end{Proposition}

The stopping rule $M^*$ given in (\ref{optiruleEPrand}) suggests that the EP would stop at $M^*$ by either transmitting or giving up the channel, which also indicates the final decision for the current round of CP. Thus, the optimal stopping rule $N^*$ could be obtained by reorganizing (\ref{optiruleEPrand}).
\begin{Proposition}\label{propCPrand}
The optimal stopping rule of CP under the i.i.d. EH model has the form as:
\begin{equation}\label{cprulerand}
   N^*=\min\left\{n\geq1: U_{M^*}(\mathcal{F}_{n,M^*})=(R_n(M^*)-\lambda^*)L\right\},
\end{equation}
where $M^*$ is the optimal stopping rule of EP given in Proposition \ref{propEPrand}. The optimal throughput $\lambda^*$ satisfies the following equation
\begin{align}\label{lambdaequirand}
\sum_{i=1}^I &Q_i\mathbb{E}\left[\mathbb{E}\left[\max\{R^i(M^{*})-\lambda^*,-\lambda^*M^{*}l/L\}\mid\mathcal{F}_{0}\right]^+\right]\nonumber\\
    =&\frac{\lambda^*l}{L}.
\end{align}
\end{Proposition}
The proof is analogous to the constant EH rate case, which is omitted here.

The next proposition, which directly follows Propositions \ref{propEPrand} and \ref{propCPrand}, concludes the overall DOS under i.i.d. EH model.
\begin{Proposition}\label{theo2}
After the $n$-th round of CP, it is optimal for the successful transmitter to take one of the following two options:
\begin{enumerate}
  \item if $\max\left\{(R_n(0)-\lambda^*)L,\mathbb{E}[U_{1}(\mathcal{F}_{n,1})\mid\mathcal{F}_{n,0}]\right\}<0$, release the channel immediately and let all transmitters start the next round of CP.
  \item otherwise, start EP following the optimal stopping rule $M_n^*$ given in Proposition \ref{propEPrand}.
\end{enumerate}
\end{Proposition}
\begin{Remark}
Propositions \ref{theo1} and \ref{theo2} summarize the DOS under the constant and i.i.d. EH models, respectively. We observe that under the constant EH model, the EP could be ``forecasted'' by finding the optimal $V^*$; then the decision of transmission would be made before starting EP. On the contrary, when the EH rates are i.i.d., such decision can only be made step by step during the EP.
\end{Remark}

\section{Battery Dynamics}\label{markov}
In this section, we validate the assumption made in Section \ref{formulation} that the energy level at each transmitter is stationary with some distribution. Firstly, we show that under the constant EH model, the energy level stored at each transmitter forms a Markov chain over time, while the state transition probabilities for different transmitters are coupled together. However, we propose an iterative algorithm to compute the corresponding steady-state distribution, which is shown converging to the global optimal point. Then, we extend our analysis to the case with i.i.d. EH rate model.

\subsection{Battery with Constant EH Model}
Note that after CP, if the successful transmitter releases the channel immediately, then the next round of CP starts, and the battery continues to be charged. If the transmitter starts the transmission, its energy level will become zero at the end of the transmission block according to Section \ref{systemmodel}. During this time, all other transmitters will keep harvesting energy within this period. Thus, the energy level transition over the transmission block can be determined. To simplify our analysis, the transmission block is treated as one time slot with length $L$ for the purpose of counting battery state transitions. In addition, we assume that the battery works in half-duplex mode, i.e., it cannot be charged when the transmitter transmits data.

For transmitter $i$ with EH rate $E^i$, $1\leq i\leq I$, the set of its energy states is given by $B_t^i\in\mathbf{\Delta}_i=\left\{0,E^il, 2E^il\cdots, \left\lfloor\frac{B_{max}\delta}{E^il}\right\rfloor E^il,B_{max}\delta\right\}$, where $t\ge1$ is the slot index. The state transition is depicted in Fig. \ref{markovchain}. In addition, we denote the distribution of the energy level for transmitter $i$ at time $t$ as $\Pi_t^i=\left[\pi^i_{t,0} \cdots\pi^i_{t,B_{max}}\right]$.
\begin{figure}
 \centering
  \includegraphics[width=3.2in]{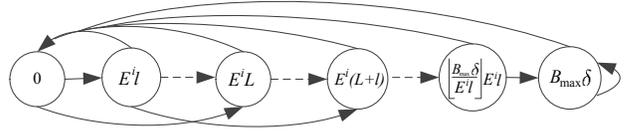}
  \caption{The state transition of the energy level at transmitter $i$ under the constant EH rate model.}
  \label{markovchain}
\end{figure}

Next, we consider the state transition probability. Suppose that transmitter $i$ is at energy level $u_i\in\mathbf{\Delta}_i$, there are three events that may happen at time slot $t$:

(i) It occupies the channel and transmits. According to Section \ref{systemmodel}, transmitter $i$ consumes all the energy for the transmission, and transfers to the energy level 0 after the transmission. Thus, the transition probability is given by
\begin{equation}\label{pu0}
  p^i_{u_i,0}=Q_ip_{tr}^i(u_i),
\end{equation}
where $Q_i$ is the probability that the $i$-th transmitter occupies the channel, and $p_{tr}^i(u_i)$ is the probability that it successfully transmits with the energy level $u_i$. Furthermore, according to (\ref{cpruleconst}), $p_{tr}^i(u_i)$ can be computed as
\begin{align}\label{ptrans}
  &p_{tr}^i(u_i)=\mathbb{P}\left\{R^i(V^*)\geq\lambda^*\right\}\nonumber\\
=&\mathbb{P}\left\{\log\left(1+|h^i|^2\frac{u_i+V^*l  E^i}{(L/l-V^*)l\sigma^2_i}\right)\geq\frac{\lambda^*}{1-\frac{V^*l}{L}}\right\},
\end{align}
where $V^{*}$ is defined by (\ref{stoppingM}) in Proposition \ref{prop2}. Note that in (\ref{ptrans}), $|h^i|^2$ is the only random variable and its distribution is known.

(ii) Other transmitters occupy the channel and transmit. If anyone among the other $I-1$ transmitters sends data, transmitter $i$ will harvest $E^iL$ units of energy during this period, and then attain level $v_i=\min\left\{u+E^iL,B_{max}\delta\right\}$. Suppose that the $j$-th transmitter transmits. Similar to the first case, the probability of transmission performed by the $j$-th transmitter is given by $Q_j\sum_{b=0}^{B_{max}}\pi^j_{t,b} p_{tr}^j(bE^jl)$, where $bE^jl\in\mathbf{\Delta}_j$ and thus $b\in\left\{0,1, 2,\cdots, \left\lfloor\frac{B_{max}\delta}{E^jl}\right\rfloor,B_{max}\right\}$. Since there are in total $I-1$ transmitters, the transition probability for the transmitter $i$ from level $u_i$ to $v_i$ is given by
\begin{equation}\label{puv1}
p^i_{u_i,v_i}=\sum_{j\neq i}Q_j\sum_{b=0}^{B_{max}}\pi^j_{t,b} p_{tr}^j(bE^jl).
\end{equation}

(iii) No transmission happens. In this case, transmitter $i$ just harvests $E^il$ units of the energy and goes into state $w_i=\min\left\{u_i+E^il,B_{max}\delta\right\}$. The probability of this case happening can be directly obtained as
\begin{equation}\label{puv2}
p^i_{u_i,w_i}=1-p^i_{u_i,0}-p^i_{u_i,v_i}.
\end{equation}
Note that when $\widetilde{u}_i=v_i=w_i$, the transition probability is just given by
\begin{align}
 p^i_{u_i,\widetilde{u}_i}&=p^i_{u_i,v_i}+p^i_{u_i,w_i}
 =p^i_{u,v_i}+1-p^i_{u_i,0}-p^i_{u_i,v_i}\nonumber\\
 &=1-p^i_{u_i,0}.\label{puv}
\end{align}

In this way, we can compute all $\{p^i_{u_i,\widetilde{u}_i}\}$ for $1\leq i\leq I$, where $u_i\in\mathbf{\Delta}_i$ and $\widetilde{u}_i\in\{0,v_i,w_i,B_{max}\delta\}$. The transition probability matrix is nothing but $\mathbf{P}^i_t=\{p^i_{u_i,\widetilde{u}_i}\}$ with dimension $\left(\left\lceil\frac{B_{max}\delta}{E^il}\right\rceil+1\right)\times\left(\left\lceil\frac{B_{max}\delta}{E^il}\right\rceil+1\right)$. Obviously, $\mathbf{P}^i_t$ is a stochastic matrix, i.e, a square matrix in which all elements are nonnegative and the row sum is 1. However, $\mathbf{P}^i_t$ depends on $t$ since $p^i_{u_i,v_i}$ depends on the state distribution $\Pi^j_t$ for all $j\neq i$. Therefore, $\left\{B_t^i\right\}_{t\geq0}$ is a non-homogeneous Markov chain, whose state evolution is given by
\begin{equation}\label{energyevo}
   \Pi^i_{t+1}=\Pi^i_t\mathbf{P}^i_t,~t\geq0.
\end{equation}

We propose Algorithm \ref{markoviter}, which is summarized in Table I,  to compute the steady-state distribution for all transmitters. Here, the infinity norm is applied, which is defined as $\parallel\mathbf{a}\parallel_{\infty}=\max_{1\leq i\leq n}|a_i|$ for $\mathbf{a}=[a_1~\cdots a_n]$.
\begin{table}[ht]
\begin{center}
\caption{Algorithm \ref{markoviter}: Compute the steady-state distribution for all transmitters.}
\hrule
\vspace{0.3cm}
\begin{itemize}
\item Initialize $\Pi^i_0$ for $1\leq i\leq I$, $\varepsilon$, and compute $p^i_{u_i,0}$ by (\ref{pu0}) for all $u_i\in\mathbf{\Delta}_i$ and $1\leq i\leq I$;
\item Set $t=0$, compute  $\mathbf{P}^i_0$ by (\ref{puv1})--(\ref{puv}) for all $1\leq i\leq I$, and compute $\Pi^i_{1}$ by (\ref{energyevo}) for all $1\leq i\leq I$. Then:
\begin{itemize}
  \item While $\max_{1\leq i\leq I}\parallel\Pi^i_{t+1}-\Pi^i_{t}\parallel_{\infty}>\varepsilon$, repeat:
  \begin{enumerate}
   \item $t=t+1$;
    \item Update $\mathbf{P}^i_t$ by (\ref{puv1})--(\ref{puv}) for all $1\leq i\leq I$;
    \item Compute $\Pi^i_{t+1}$ by (\ref{energyevo}) for all $1\leq i\leq I$;
 \end{enumerate}
  \item end.
\end{itemize}
\item Algorithm ends.
\end{itemize}
\vspace{0.2cm} \hrule \label{markoviter} \end{center}
\end{table}

\begin{Proposition}\label{algorithmmarkov}
For any given initial state distribution $\Pi^i_0$, $\Pi^i_t=\left[\pi^i_{t,0} \cdots\pi^i_{t,B_{max}}\right]$ that is  generated by Algorithm \ref{markoviter}, converges to a unique steady-state distribution $\Pi^i$ for all $1\leq i\leq I$.
\end{Proposition}
The proof is given in Appendix D.
\begin{Remark}
The steady-state distribution for all transmitters can be obtained by the iterative computation $\mathbf{\Pi}_{t+1}=\mathbf{\Pi}_t\mathbf{P}$ over the ``super'' Markov system as well, which is constructed in Appendix D. However, this is not as efficient as Algorithm \ref{markoviter}. From the computational complexity point of view, suppose that each transmitter has $m$ energy levels, and there are $n$ transmitters in total. The number of the states in the ``super'' Markov chain is $m^n$. If there is only one processer, the floating-point calculation for one iteration of the state distribution for the ``super'' Markov chain is approximately on the order of $O\left(2m^{2n}\right)$. On the contrary, by using Algorithm \ref{markoviter}, (\ref{puv1}) requires $n^2m^2$ calculations, and updating $\{\mathbf{P}^i_{t}\}$ requires about $nm$ calculations according to (\ref{puv2}). In addition, $\{\Pi^i_{t}\mathbf{P}^i_{t}\}$ requires $2nm^2$ calculations. Overall, one iteration for all transmitters is approximately on the order of $O\left(n^2m^2\right)$, which is more efficient than the case for the ``super'' Markov chain especially when $m$ and $n$ are large. Moreover, our algorithm can also be operated in a parallel way, i.e., computing $\Pi^i_{t+1}=\Pi^i_{t}\mathbf{P}^i_{t}$ for $1\leq i\leq n$ at the same time over different cores.
\end{Remark}

\subsection{Battery with i.i.d. EH Model}
The argument that the battery state evolves as a Markov process for the random case is analogous to that of the constant case in the previous subsection. The main difference is that the probability $p_{tr}^i(u_i)$ defined by (\ref{ptrans}) is changed, which needs to be further developed under the i.i.d. EH rate model.

We now consider the calculation of $p_{tr}^i(u_i)$. When transmitter $i$ grabs the channel with energy level $u_i$, according to the stopping rule $M^*$ (\ref{optiruleEPrand}) and $N^*$ (\ref{cprulerand}), the transmitter checks the condition $\max\left\{(R(0)-\lambda)L,-\lambda l+\mathbb{E}[U_{1}(\mathcal{F}_{1})\mid\mathcal{F}_{0}]\right\}\geq0$. If it is true, the transmitter starts EP until the $M^*$-th slot and transmits when $(R(M^*)-\lambda^*)L\geq-\lambda^*M^*l$ according to (\ref{cprulerand}). Specifically, given $U_0(u_i,|h^i|^2)\geq0$, the transmitter continues EP at slot $k$ for $0\leq k\leq M^*-1$, which is equivalent to $\max\{(R(k)-\lambda^*)L,-\lambda^*kl\}<\mathbb{E}[U_{k+1}(\mathcal{F}_{k+1})\mid\mathcal{F}_{k}]$, where $\mathcal{F}_{k}=\{u_i+l\sum_{j=0}^kE_j^i,|h^i|^2\}$. Then, at slot $M^*=m\le L/l$, the transmitter stops EP and transmits when $(R(m)-\lambda^*)L\geq\max\{-\lambda^*ml,\mathbb{E}[U_{m+1}(\mathcal{F}_{m+1})\mid\mathcal{F}_{m}]\}$. Thus, we obtain
\begin{align}\label{probtranrand}
    p_{tr}^i(u_i)&=\int_{0}^{\infty}\mathbb{P}\left\{\hbox{Transmits at $M^*$}\mid U_0(u_i,d|h^i|^2)\geq0\right\}\cdot\nonumber\\
&~~~~~~~~~\mathbb{P}\left\{U_0(u_i,d|h^i|^2)\geq0\right\}f(|h^i|^2)d|h^i|^2,
\end{align}
where $f(|h^i|^2)$ is the probability density function (PDF) of the channel power gain. The probability $\mathbb{P}\left\{U_0(u_i,d|h^i|^2)\geq0\right\}$ can be computed based on Proposition \ref{propEPrand}. For notation simplicity, we omit the condition $U_0(u_i,d|h^i|^2)\geq0$, and the first term in the integral of (\ref{probtranrand}) can be expanded as
\begin{align}
\mathbb{P}\left\{\hbox{Transmits at $M^*$}\right\}=\sum_{m=0}^{L/l}\left(\prod_{k=0}^{m-1}\mathbb{P}\left\{\alpha_k<0\right\}\right)\mathbb{P}\left\{\beta_m\leq0\right\}\label{tempprobtran}
\end{align}
where $\alpha_k=\max\{(R(k)-\lambda^*)L,-\lambda kl\}-\mathbb{E}[U_{k+1}(\mathcal{F}_{k+1})\mid \mathcal{F}_{k}]$, and $\beta_m=\max\{-\lambda ml,\mathbb{E}[U_{m+1}(\mathcal{F}_{m+1})\mid \mathcal{F}_{m}]\}-(R(m)-\lambda^*)L$. Note that in $\mathbb{P}\left\{\alpha_k<0\right\}$, $R(k)$ and $\mathbb{E}[U_{k+1}(\mathcal{F}_{k+1})\mid \mathcal{F}_{k}]$ are random since they are the functions of $\sum_{j=0}^kE^i_j$, where $\left\{E_j^i\right\}_{1\leq j\leq k}$ are i.i.d. with a known distribution and $E_0^i=0$. Thus, $\mathbb{P}\left\{\alpha_k<0\right\}$ can be computed. Using the similar argument, it is easy to see that $\mathbb{P}\left\{\beta_m\leq0\right\}$ can be computed as well. Therefore, the probability given in (\ref{tempprobtran}) is computable. Overall, we could obtain $p_{tr}^i(u_i)$ after plugging (\ref{tempprobtran}) into (\ref{probtranrand}).

After obtaining $p_{tr}^i(u_i)$, the transition probability $\{p^i_{u_i,\widetilde{u}_i}\}$, where $u_i\in\mathbf{\Delta}$, and $\widetilde{u}_i\in\{0,u_i,u_i+\delta,\cdots,B_{max}\delta\}$, can be calculated similarly as the case of constant EH rate. In addition, Algorithm \ref{markoviter} and Proposition \ref{algorithmmarkov} could be modified, such that they could suit the i.i.d. EH model, which is omitted in this paper.

\section{Computation of the Optimal Throughput}\label{onedsearch}
The optimal throughput $\lambda^*$ hinges upon the optimal stopping rules in (\ref{cpruleconst}) and (\ref{cprulerand}). Thus, to fully obtain the optimal scheduling policy of the proposed DOS, we next turn our attention to computing the value of $\lambda^*$.

By Propositions \ref{stopingruleN} and \ref{propCPrand}, $\lambda^*$ can be obtained by solving (\ref{lambdaequi}) or (\ref{lambdaequirand}) under the constant or i.i.d. EH model, respectively. Next, we briefly introduce the idea why there exists $\lambda^*$ such that the equation (\ref{lambdaequi}) or (\ref{lambdaequirand}) holds, and how to search $\lambda^*$. For brevity, we focus the constant EH rate case.

Note that $R(V^*)$ is a function of random variables $h^i$ and $B_{0}^i$; we could calculate the expectation on the left-hand side of (\ref{lambdaequi}) for each given $\lambda\geq0$. Such expectation requires the distribution of $B_{0}^i$, i.e., the steady-state distribution $\Pi^i$, which could be approximately computed as shown in Section \ref{markov}. In addition, for a given $\lambda$, an upper bound of this expectation can be obtained by fixing $\Pi^i=[0,~\cdots 0,~1]$. As $\lambda$ increases from zero to infinity, this upper bound decreases to zero at some $\widetilde{\lambda}<\infty$. Since the right-hand side of (\ref{lambdaequi}) is strictly increasing over $\lambda$ within the range $[0,+\infty)$, there at least exists one $\lambda^*$ satisfying (\ref{lambdaequi}). Therefore, an exhaustive one-dimension search can be applied to obtain the optimal throughput over the range $\left[0,\widetilde{\lambda}\right]$. Note that during each iteration of the exhaustive search, Algorithm \ref{markoviter} (given in Section \ref{markov}) is used to obtain the steady-state distribution for a given $\lambda\in\left[0,\widetilde{\lambda}\right]$, and then we check if the equation (\ref{lambdaequi}) or (\ref{lambdaequirand}) holds. Finally, $\lambda^*$ should be the largest one in $\left[0,\widetilde{\lambda}\right]$ that makes the equation (\ref{lambdaequi}) or (\ref{lambdaequirand}) hold.

In summary, the above search can characterize the optimal stopping rules given in Propositions \ref{stopingruleN} and \ref{propCPrand}, which completes the proposed DOS framework.

\section{Numerical results}\label{numerical}
In this section, we first validate Propositions \ref{stopingruleN} and \ref{propCPrand} to show that the optimal throughput $\lambda^*$ exists and can be found via one-dimension search. Second, we investigate the throughput gain of our proposed DOS with two-level probing over the best-effort delivery method, where the data is transmitted whenever the channel contention is successful. Note that such a method can be realized in the proposed DOS framework by fixing $M=0$ and setting $N=1$ in (\ref{cpruleconst}) and (\ref{cprulerand}). Let $\lambda_0$ denote the throughput obtained by the best-effort scheme, which can be calculated as
\begin{equation}\label{lambdazero}
    \lambda_0=\frac{\sum_{i=1}^I\frac{Q_i}{Q}\mathbb{E}\left[L\log\left(1+|h_n^i|^2\frac{B_{n,0}^i}{L\sigma^2}\right)\right]}{\frac{l}{Q}+L}.
\end{equation}

In general, a typical button cell battery has the capacity of 150~mAh with the end-point voltage of 0.9~V, which is equal to 150~mAh $\times$ 3600~s/h  $\times$ 0.9~V = 486~J. A thin-film rechargeable battery can offer 50~$\mu$Ah with 3.3~V, which is equal to 0.594~J. Since a typical transmission time interval is on the time scale of milliseconds, we let the energy unit be $\delta=10^{-3}$~J in the simulation. Accordingly, we set the capacity of the battery $B_{max}\delta=10^5\delta$, which falls between the capacity volume of a thin-film battery and that of a button cell battery. Also, the current commercial solar panel can provide power from 1~W to about 400~W, which is equivalent to $1\delta\cdot$ms$^{-1}$ $\sim$ $400\delta\cdot$ms$^{-1}$. According to this fact, in our simulation, we let the EH rate vary within the range $[0,40\delta]$. In addition, the channel gains are i.i.d for different links and the channel power gains follow an exponential distribution with mean 5. The variance of the noise is set to be 10~mW. The length of one time slot is unified as $l=1$~ms and the length of a transmission block is $L=100l$.

\begin{figure}
  \centering
  \includegraphics[width=3.2in]{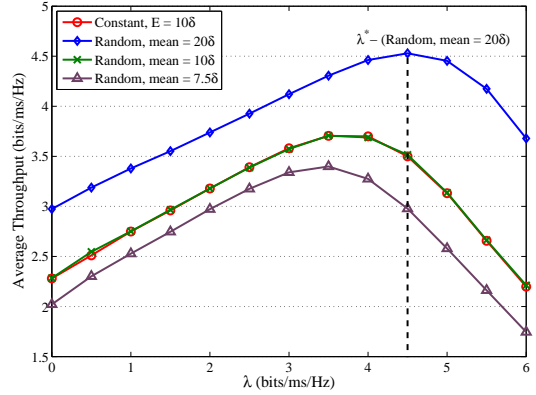}
  \caption{$\lambda$ v.s. the average throughput.}
  \label{lamdvsthrough}
\end{figure}
\subsubsection{Validation of Propositions \ref{stopingruleN} and \ref{propCPrand}}
In Fig.~\ref{lamdvsthrough}, we illustrate the variation of the average throughput as the ``threshold'' $\lambda$ changes. Without loss of generality, we first consider a homogeneous network with $10$ user pairs, i.e., all pairs are identical. For the constant EH model, the EH rate is set to be $E=10\delta$ for all transmitters. For the i.i.d. EH case, we choose the Bernoulli model \cite{AS,MK}: The EH rate is either zero or of a finite value with probability 0.5. In our simulation, we consider three cases for the mean values in i.i.d. EH model: $7.5\delta$, $10\delta$, and $20\delta$.

First, we observe in Fig.~\ref{lamdvsthrough} that as $\lambda$ increases from zero, the average throughput is increasing then decreasing. Then, the optimal point is achieved at $\lambda^*$, where the average throughput is at its apex that is also approximately of the same value as $\lambda^*$. Taking the case of i.i.d. EH model with mean $20\delta$ as an example in Fig.~\ref{lamdvsthrough}, the value of the optimal throughput $\lambda=\lambda^*$ is approximately 4.5, and the actual optimal average throughput is about 4.5 as well. Therefore, this observation validates our Propositions \ref{stopingruleN}, \ref{propCPrand} and discussions in Section \ref{onedsearch}. Second, we observe that the average throughput is almost the same when the mean of the EH rate in the i.i.d. EH model is equal to the EH rate in the constant EH model. Thus, the type of EH rate models does not directly determine the average throughput performance.

\subsubsection{Throughput gain}

We use $\lambda_{EP}$ to denote the throughput where only EP is adopted, i.e., setting $N=1$ and $M=M^*$, and $\lambda_{CP}$ to denote the throughput where only CP is adopted, i.e., setting $N=N^*$ and $M=0$. Thus, the throughput gains are defined as:
\begin{equation}
\left\{
  \begin{array}{ll}
    G_{EP}=\frac{\lambda_{EP}-\lambda_0}{\lambda_0}, & \hbox{gain from EP;} \\
    G_{CP}=\frac{\lambda_{CP}-\lambda_0}{\lambda_0}, & \hbox{gain from CP;} \\
    G_{DOS}=\frac{\lambda^*-\lambda_0}{\lambda_0}, & \hbox{gain from CP + EP.}
  \end{array}
\right.
\end{equation}

\begin{figure}
  \centering
  \includegraphics[width=3.2in]{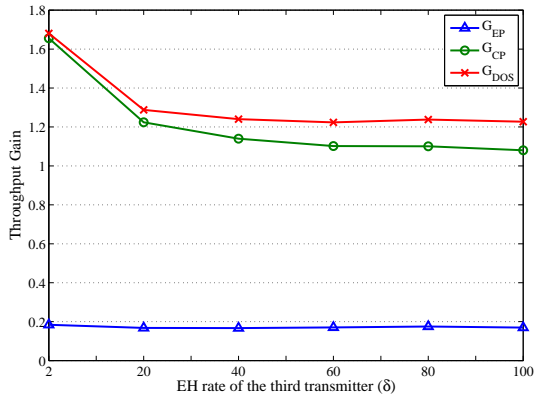}
  \caption{The throughput gain v.s. EH rate of the third transmitter.}
  \label{thgain}
\end{figure}
In Fig.~\ref{thgain}, we evaluate the above throughput gains for the network with $I=3$ user pairs. Recall from Section \ref{systemmodel} that our analysis is applicable for $I\geq 2$. Since the constant and i.i.d. EH rate models could attain the same throughput performance over $\lambda$, we only consider the constant EH model in this case. Particularly, we study a heterogeneous case where the first two transmitters have the same EH rates $2\delta$, while the EH rate of the third transmitter varies from $2\delta$ to $100\delta$.

We observe in Fig.~\ref{thgain} that as the EH rate of the third transmitter increases, $G_{EP}$ almost keeps constant and can achieve a gain about 19\%. It implies that after the channel contention, the successful transmitter with any EH rate could do EP to enhance its average transmission rate over the transmission block. Thus, the ESI of the successful transmitter does not have obvious impact on the throughput. However, we notice that $G_{CP}$ achieves its maximum when all transmitters are identical (with the same EH rate $2\delta$) and then decreases slowly as the EH rate of the third transmitter increases. The intuition is that when the difference among EH rates becomes larger, the stopping rule of CP will more likely let the transmitter with relatively low energy level to give up the channel, which results in a longer time on CP and then the throughput gain is lower than the case when all transmitters are identical. Regarding $G_{DOS}$, our proposed DOS with two-stage probing can achieve the highest throughput gain among three schemes. It is worth noticing that as the EH rate of the third transmitter increases, the efficiency of DOS becomes more apparent, although slowly, than the scheme with pure CP, which implies that the second stage probing brings more benefits. Our intuition is that a larger difference among the EH rates leads to a bigger difference of energy levels. Since EP allows the successful transmitter with relatively lower energy level to possibly harvest more energy after CP, EP will plays a more important role as the difference among the EH rates increases.

\begin{figure}
  \centering
  \includegraphics[width=3.2in]{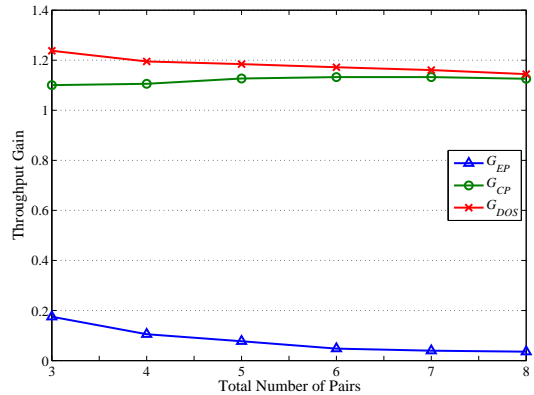}
  \caption{The throughput gain v.s. the size of the network.}
  \label{thgainoversize}
\end{figure}
In Fig.~\ref{thgainoversize}, we illustrate how the size of the network influences the throughput gains. In this scenario, we start from a three-pair network with EH rates $2\delta$, $2\delta$, and $80\delta$, respectively. Then, we keep adding pairs with EH rate $2\delta$ at the transmitter side. We observe that the throughput gain $G_{CP}$ is increasing a little as the size of the network is increasing. It is reasonable since CP could utilize the multi-user diversity of both channel gains and energy levels. We see that $G_{CP}$ increases slowly, since we only add a low-EH-rate transmitter at each time. We also observe that $G_{EP}$ is decreasing. The reason is that the more transmitters in the network, the less probability to transmit for each transmitter, and then more transmitters would maintain a high energy level. Thus, EP is rarely triggered after a channel contention. For the same reason, $G_{DOS}$ would approach $G_{CP}$ as the size of the network increases.

\section{Conclusion}\label{fin}
In this paper, we proposed a DOS framework for a heterogeneous single-hop \emph{ad hoc} network, in which each transmitter is powered by a renewable energy source and accesses the channel randomly. Our DOS framework includes two successive processes: All transmitters first probe the channel via random access, and then the successful transmitter decides whether to give up the channel or to optimally probe the energy before data transmission. The optimal scheduling policy of the DOS framework is obtained as follows: First, assuming the battery state is stationary at each transmitter, the expected throughput maximization problem was formulated as a rate-of-return optimal stopping problem, which was solved for both the constant and i.i.d. EH rate models; second, by fixing the stopping rule, the stored energy level at each transmitter was shown to own a steady-state distribution as time goes to infinity, where we also proposed an efficient iterative algorithm for its computation; finally, the optimal throughput and the scheduling policy is obtained via one-dimension search with the above two steps (i.e., finding the form of the optimal stopping rule and calculating the steady-state distribution) repeated in each iteration. Numerical results were also provided to validate our analysis; the proposed DOS with two-level probing was shown to outperform the best-effort delivery method.


\section*{Appendices}
\subsection{Proof of Proposition \ref{lemma1}}
For the first part of Proposition \ref{lemma1}, it follows by Theorem 1 in Chapter 3 of \cite{optstop} that $N^*(\lambda)$ exists and $S^*(\lambda)$ is attained by this $N^*(\lambda)$ if the following two conditions are satisfied:
\begin{description}
  \item[(C1)] $\limsup_{N\rightarrow\infty}r_N(\lambda)$ $\leq r_{-\infty}(\lambda)$ a.s.;
  \item[(C2)] $\mathbb{E}\left[\sup_{N\geq1} r_N(\lambda)\right]<\infty$,
\end{description}
where $r_N(\lambda)$ is given by (\ref{rnlambda}). As we pointed out in Section \ref{systemmodel}, the energy level $B_{N,0}$ is stationary for $N\geq1$. Although $\{R_N(M_{N}^*)\}_{N\geq1}$ are independent, it may not be identically distributed with respect to $h_N$ and $B_{N,0}$. However, it is not too difficult to show that (C1) and (C2) hold. The idea is that we first consider that every transmitter has the same statistics; then we apply the channel contention probability as the summation coefficients over all transmitters.

For (C1), if we assume that all transmitters have the same statistics as transmitter $i$, then $\{R_N^i(M_{N}^*)\}_{N\geq1}$ become i.i.d.. Since $\mathbb{E}\left[R_{N}^i(M_N^*)\right]<\infty$ according to Section \ref{formulation}, and the accumulated cost $\lambda T_N=\lambda l\left(K_N+\sum_{n=1}^{N-1}(K_n+M_n^{*})\right)\rightarrow\infty$ as $N\rightarrow\infty$ a.s., we obtain that $\mathbb{P}\left\{\limsup_{N\rightarrow\infty}r_N^i(\lambda)=-\infty\right\}=1$. Recall from Section \ref{systemmodel} that the channel is occupied by transmitter $i$ with probability $Q_i$ and $\sum_{i=1}^I\frac{Q_i}{Q}=1$, we obtain that
\begin{align}
1&=\sum_{i=1}^I\frac{Q_i}{Q}\mathbb{P}\left\{\limsup_{N\rightarrow\infty}r_N^i(\lambda)=-\infty\right\}\nonumber\\
&=\mathbb{P}\left\{\limsup_{N\rightarrow\infty}r_N(\lambda)=-\infty\right\},\nonumber
\end{align}
which proves that (C1) holds.

For (C2), it can be shown that
\begin{align}\label{lastterm}
\mathbb{E}&\left[\sup_{N\geq1} r_N^i(\lambda)\right]=\mathbb{E}\left[\sup_{N\geq1} \left(\left(R_{N}^i(M_N^*)-\lambda\right)L-\lambda T_N\right)\right]\nonumber\\
&~~~~~~~~~~~~~\leq\mathbb{E}\left[\sup_{N\geq1}\left(R_{N}^i(M_N^*)-\lambda (lN+L)\right)\right],
\end{align}
due to the fact that $K_n\geq1$ and $M_n^{*}\geq0$ for $1\leq n\leq N$. Since $\mathbb{E}\left[\left(R_N^i(M_N^{*})\right)^2\right]<\infty$, it follows that the right-hand side of (\ref{lastterm}) is finite by Theorem 1 in Chapter 4 of \cite{optstop}. Similar to the technique in the proof of (C1), we have
\begin{align}
\mathbb{E}\left[\sup_{N\geq1}r_N(\lambda)\right]&=\sum_{i=1}^I\frac{Q_i}{Q}\mathbb{E}\left[\sup_{N\geq1}r_N^i(\lambda)\right]<\infty,\nonumber
\end{align}
which shows that (C2) also holds.


For the second part, we know that with the cost $\lambda lK_N$ at the $N$-th CP for any $N\geq1$, the successful transmitter could choose one of three actions: transmits immediately with reward $(R_N(0)-\lambda)L$; or gives up the channel immediately, and obtains the optimal expected net reward $S^*(\lambda)$ based on the property of time invariance described in Section \ref{subformulation}; or starts EP and obtains the expected net reward $\mathbb{E}\left[U_1(\mathcal{F}_{N,1})\mid\mathcal{F}_{N,0}\right]$. Thus, by the optimal stopping theory \cite{optstop,somepro}, $S^*(\lambda)$ satisfies the \emph{optimality equation} under (C2) as
\begin{align}
&S^*(\lambda)=-\lambda lK_N+\nonumber\\
&\max\left\{S^*(\lambda),(R_N(0)-\lambda)L,\mathbb{E}\left[U_1(\mathcal{F}_{N,1})\mid\mathcal{F}_{N,0}\right]\right\},\nonumber
\end{align}
which is equivalent to (\ref{optiequi}).

\subsection{Proof of Proposition \ref{lemma2}}
For 1), we show the concavity of function $y(x)$ by checking its second-order derivative over $[0,1)$, which is given by
\begin{equation}
    y''(x)=-\frac{(a+b)^2}{(1-x)\left[a+1+(b-1)x\right]^2}\leq0. \nonumber
\end{equation}
Therefore, $y(x)$ is concave over $[0,1)$ \cite{boyd}. To prove the second part of 1), we check the first-order derivative of $y(x)$, which is given by
\begin{equation}\label{firstorder}
    y'(x)=-\log\left(1+\frac{a+bx}{1-x}\right)+\frac{a+b}{1-x+a+bx}.
\end{equation}
It is easy to see that as $x\rightarrow1^-$, the first term of the right-hand side of (\ref{firstorder}) goes to negative infinity, while the second term is bounded. Hence, $y'(x)$ is strictly negative as $x\rightarrow1^-$. Therefore, part 1) is proved.

Next, we prove 2). By checking the second-order derivative of $g(x)$, we obtain
\begin{equation}
  g''(x)=-\frac{a^2}{(1-x)(a+1-x)^2}\leq0,\nonumber
\end{equation}
which implies that $g(x)$ is concave. For the second part of 2), we consider the first-order derivative of $g(x)$, which is given by
\begin{equation}\label{firstorderg}
    g'(x)=-\log\left(1+\frac{a}{1-x}\right)+\frac{a}{1-x+a}.
\end{equation}
Since $g''(x)\leq0$, it follows that
\begin{equation}
    \max_{0\leq x<1}g'(x)=g'(0)=-\log\left(1+a\right)+\frac{a}{1+a}.\nonumber
\end{equation}
Moreover, due to the fact that $\frac{\mathrm{d}}{\mathrm{d}a}\left(-\log\left(1+a\right)+\frac{a}{1+a}\right)=-\frac{a}{(1+a)^2}\leq0$ for arbitrary $a\geq0$, we obtain
\begin{equation}
    \max_{0\leq x<1}g'(x)=g'(0)\leq\left.\left(-\log\left(1+a\right)+\frac{a}{1+a}\right)\right|_{a=0}=0,\nonumber
\end{equation}
which proves the second part of 2).

\subsection{Proof of Proposition \ref{prop1}}
According to Part 1) of Proposition \ref{lemma2}, we obtain that $G(\rho)$ is concave over $\rho\in[0,1)$, which means that $G'(\rho)=\frac{\mathrm{d}G(\rho)}{\mathrm{d}\rho}$ is decreasing over $[0,1)$ and attains its maximum at $\rho=0$.
Then, finding the maximum of $G(\rho)$ boils down to two cases:
\begin{enumerate}
  \item $\left.G'(\rho)\right|_{\rho=0}<0$: It follows that $G(\rho)$ is decreasing over $[0,1)$, and $\rho^*=0$ is the optimum.
  \item $\left.G'(\rho)\right|_{\rho=0}\geq0$: The point $\rho_0$, satisfying $\left.G'(\rho)\right|_{\rho=\rho_0}=0$, lies on the right-hand side of $\rho=0$. By Part 1) of Proposition \ref{lemma2}, $G'(\rho)<0$ as $\rho\rightarrow1^-$, which implies that $\rho_0\in[0,1)$. Since the optimal point $\rho^*\leq\frac{B_{max}\delta-B_{0}}{LE}$ due to (\ref{optig}), it follows that $\rho^*=\min\left\{\rho_0,\frac{B_{max}\delta-B_{0}}{LE}\right\}$.
\end{enumerate}
Note that $\left.G'(\rho)\right|_{\rho=0}\geq0$ is equivalent to $\frac{C+D}{1+C}\geq\log(1+C)$, where $C=\frac{|h|^2B_{0}}{L\sigma^2}\geq0$, $D=\frac{|h|^2E}{\sigma^2}\geq0$, and $\left.G'(\rho)\right|_{\rho=\rho_0}=0$ is equivalent to
\begin{equation}\label{rho0}
   \log\left(1+\frac{C+D\rho_0}{1-\rho_0}\right)=\frac{C+D}{1-\rho_0+C+D\rho_0}.
\end{equation}

Next, we show that when $\frac{C+D}{1+C}\geq\log(1+C)$, (\ref{rho0}) has a unique solution.  For $\rho\in[0,1)$, the left-hand side of (\ref{rho0}) is increasing over $\rho$ from $\log\left(1+C\right)$ to $+\infty$. For its right-hand side, we have the following two cases:
\begin{enumerate}
  \item $D\geq1$: The right-hand side of (\ref{rho0}) decreases from $\frac{C+D}{1+C}$ to $1$. Since $\frac{C+D}{1+C}\geq\log(1+C)$, there exists a unique solution $\rho_0$ for (\ref{rho0});
  \item $0\leq D<1$: The right-hand side of (\ref{rho0}) increases from $\frac{C+D}{1+C}$ to $1$.
If the first-order derivative of the left-hand side of (\ref{rho0}) is always greater than that of the right-hand side, there must be only one solution for (\ref{rho0}) when $\frac{C+D}{1+C}\geq\log(1+C)$. Thus, we check their first-order derivatives: For the left-hand side of (\ref{rho0}), we obtain
\begin{align}\label{left}
&\frac{\mathrm{d}}{\mathrm{d}\rho}\log\left(1+\frac{C+D\rho}{1-\rho}\right)=\frac{C+D}{(1-\rho)\left(1+C+(D-1)\rho\right)};
\end{align}
for the right-hand side, we have
\begin{align}\label{right}
 \frac{\mathrm{d}}{\mathrm{d}\rho}\left(\frac{C+D}{1-\rho+C+D\rho}\right)=\frac{(C+D)(1-D)}{\left(1+C+(D-1)\rho\right)^2}.
\end{align}
Thus, by calculating the difference between (\ref{left}) and (\ref{right}), we arrive at
\begin{align}
&\frac{C+D}{(1-\rho)\left(1+C+(D-1)\rho\right)}-\frac{(C+D)(1-D)}{\left(1+C+(D-1)\rho\right)^2}\nonumber\\
=&\frac{(C+D)^2}{(1-\rho)\left(1+C+(D-1)\rho\right)^2}\geq0.
\end{align}
Therefore, there exists a unique solution $\rho_0$ satisfying (\ref{rho0}).
\end{enumerate}
In conclusion, the proposition is proved.

\emph{Remark}: Since it is proved that $\rho_0$ is unique in (\ref{rho0}), $\rho_0$ can be found just by adopting a simple one-dimension searching method, e.g., bisection search.

\subsection{Proof of Proposition \ref{algorithmmarkov}}
To prove this proposition, we construct an axillary ``super'' Markov chain in which each state is a ``super'' vector of aggregated energy levels across the whole network, whose transition probability matrix does not change over time $t$. Afterwards, we prove that such a ``super'' Markov chain has a unique steady-state distribution. Then, we show that for any time $t$ in the original Markov chain, one iteration for updating $\Pi^i_{t}$ for $1\leq i\leq I$ in Algorithm \ref{markoviter}  is equivalent to the evolution of the state distribution in the ``super'' Markov chain, thereby proving the convergence of Algorithm \ref{markoviter}.

To construct such a ``super'' Markov chain, we need to jointly consider the states of energy levels across all transmitters. Let $\mathbf{\Sigma}$ denote the set of all possible battery states over the whole system, i.e.,
\begin{equation}\label{sigmaset}
 \mathbf{\Sigma}=\left\{\mathbf{u}=\left(u_1~\cdots~u_I\right): u_1\in\mathbf{\Delta}_1,\cdots, u_I\in\mathbf{\Delta}_I \right\}.
\end{equation}
Furthermore, we use $\mathbf{B}_t$ to denote the battery state of the system at time $t$, and thus we have $\mathbf{B}_t\in\mathbf{\Sigma}$. Note that the number of elements in $\mathbf{\Sigma}$ is $\left(\left\lceil\frac{B_{max}\delta}{E^1l}\right\rceil+1\right)\times\cdots\times\left(\left\lceil\frac{B_{max}\delta}{E^Il}\right\rceil+1\right)$.

Suppose that $\mathbf{B}_t=\mathbf{u}$. There are $I+1$ possible events at time $t$: A transmission is performed by transmitter $i$, where $1\leq i\leq I$, or no transmission happens.

If the $i$-th transmitter transmits, there is $ \mathbf{B}_{t+1}=\mathbf{v}_i$, where $\mathbf{v}_i\in\mathbf{\Sigma}$ and
\begin{equation}
\mathbf{v}_i=\left(
     \begin{array}{ll}
       \min\{u_1+E^1L, B_{max}\delta\} \\
       \cdots \\
       0 \\
       \cdots \\
       \min\{u_I+E^IL, B_{max}\delta\}
     \end{array}
   \right)^T,\nonumber
\end{equation}
in which the $i$-th element is zero. According to (\ref{pu0}), the corresponding transition probability is given by
\begin{equation}\label{puvi}
 p_{\mathbf{u},\mathbf{v}_i}=Q_ip_{tr}^i(u_i), ~1\leq i\leq I.
\end{equation}

If no transmission happens, all transmitters just harvest energy for one time slot. Then, we obtain $\mathbf{B}_{t+1}=\mathbf{w}$, where $\mathbf{w}\in\mathbf{\Sigma}$ and
\begin{equation}
\mathbf{w}=\left(
     \begin{array}{ll}
       \min\{u_1+E^1l, B_{max}\delta\} \\
       \cdots \\
       \min\{u_i+E^il, B_{max}\delta\} \\
       \cdots \\
       \min\{u_I+E^Il, B_{max}\delta\}
     \end{array}
   \right)^T.\nonumber
\end{equation}
The corresponding transition probability is just the complement of the transmission probability over all other possible $I$ cases, which is given by
\begin{equation}\label{puw}
 p_{\mathbf{u},\mathbf{w}}=1-\sum_{i=1}^IQ_ip_{tr}^i(u_i).
\end{equation}
Therefore, $\{\mathbf{B}_t\}_{t\geq0}$ is a unichain \cite{Gallager}, i.e., a finite-state Markov process that contains a single recurrent class. By calculating the transition probability for each $\mathbf{u}\in\mathbf{\Sigma}$, we obtain the transition probability matrix $\mathbf{P}$ for $\{\mathbf{B}_t\}_{t\geq0}$. Clearly, $\mathbf{P}$ is a stochastic matrix and is invariant over time. Therefore, there exists a unique probability vector $\mathbf{\Pi}$ such that $\mathbf{\Pi}=\mathbf{\Pi}\mathbf{P}$ holds \cite{Gallager}. In fact, $\mathbf{\Pi}$ is the steady-state distribution of $\{\mathbf{B}_t\}_{t\geq0}$.

So far, we have constructed a ``super'' Markov chain $\{\mathbf{B}_t\}_{t\geq0}$ for the whole system, for which the steady-state distribution exists and is unique. Therefore, by the iteration $\mathbf{\Pi}_{t+1}=\mathbf{\Pi}_t\mathbf{P}$, we have $\lim_{t\rightarrow\infty}\mathbf{\Pi}_t=\mathbf{\Pi}$. Thus, it suffices to show that
\begin{equation}\label{equivalence}
   \mathbf{\Pi}_{t+1}=\mathbf{\Pi}_t\mathbf{P}~\Leftrightarrow~\left\{
     \begin{array}{ll}
       \Pi^1_{t+1}=\Pi^1_{t}\mathbf{P}^1_{t}, \\
       \cdots \\
       \Pi^i_{t+1}=\Pi^i_{t}\mathbf{P}^i_{t}, \\
       \cdots \\
       \Pi^I_{t+1}=\Pi^I_{t}\mathbf{P}^I_{t}.
     \end{array}
   \right.~~~t\geq0,
\end{equation}
If (\ref{equivalence}) is true, the state distribution of each transmitter converges to the unique steady-state distribution.

Next, we are going to show that both the directions ``$\Rightarrow$'' and ``$\Leftarrow$'' of (\ref{equivalence}) hold. For notational simplicity, we omit the time index $t$. In fact, the direction ``$\Leftarrow$'' is the same as constructing the ``super'' Markov chain as discussed earlier. If the system is at state $\mathbf{u}=\left(b_1E^1l \cdots b_IE^Il\right)$, where $b_i\in\left\{0,1,2,\cdots, \left\lfloor\frac{B_{max}\delta}{E^il}\right\rfloor,B_{max}\right\}$, $1\leq i\leq I$, the probability $\mathbf{\Pi}(\mathbf{u})$ is the joint probability over all transmitters, i.e., $\mathbf{\Pi}(\mathbf{u})=\prod_{i=1}^I\pi^i_{b_i}$. The way of constructing transition probability matrix $\mathbf{P}$ is given by (\ref{puvi}) and (\ref{puw}), which can be obtained directly from (\ref{pu0}) for $\{\mathbf{P}^i\}$. Thus, both  $\mathbf{\Pi}$ and $\mathbf{P}$ can be obtained from the right-hand side of (\ref{equivalence}).

For the direction ``$\Rightarrow$'' of (\ref{equivalence}), we need to show how we obtain $\{\Pi^i\}$ and $\{\mathbf{P}^i\}$ from the left-hand side of (\ref{equivalence}). We consider $\{\Pi^i\}$ first. Given the state distribution $\mathbf{\Pi}$ of the system, there exists an one-to-one mapping from each element of $\mathbf{\Sigma}$ to that of $\mathbf{\Pi}$. Let $\mathbf{\Pi}(\mathbf{u})$ denote the probability of the system staying at state $\mathbf{u}\in\mathbf{\Sigma}$. Obviously, there is $\sum_{\mathbf{u}\in\mathbf{\Sigma}}\mathbf{\Pi}(\mathbf{u})=1$. Then, we consider the subset of $\mathbf{\Sigma}$ such that transmitter $i$ stays at state $u\in\mathbf{\Delta}_i$, i.e.,
\begin{align}\label{sigmasetfori}
\mathbf{\Sigma}_{u_i=u}=&\left\{\mathbf{u}=\left(u_1 \cdots u_i \cdots  u_I\right):\right.\nonumber\\
& \left.u_1\in\mathbf{\Delta}_1,\cdots,u_i=u,\cdots, u_I\in\mathbf{\Delta}_I \right\}.
\end{align}
Clearly, (\ref{sigmasetfori}) satisfies $\bigcup_{u\in\mathbf{\Delta}_i}\mathbf{\Sigma}_{u_i=u}=\mathbf{\Sigma}$. Then, the probability that transmitter $i$ stays at state $u=bE^il$, where $b\in\left\{0,1, 2,\cdots, \left\lfloor\frac{B_{max}\delta}{E^il}\right\rfloor,B_{max}\right\}$, is equal to the probability that the system is staying at $\mathbf{\Sigma}_{u_i=u}$, i.e.,
\begin{equation}\label{bdelta}
\pi_{b}^i=\mathbb{P}\left\{\mathbf{\Sigma}_{u_i=u}\right\}=\sum_{\mathbf{u}\in\mathbf{\Sigma}_{u_i=u}}\mathbf{\Pi}(\mathbf{u}).
\end{equation}
In this way, we can obtain the state distribution $\Pi^i$ for transmitter $i$ such that $\Pi^i=[\pi_{0}^i \cdots \pi_{b}^i\cdots\pi_{B_{max}}^i]$.

Next, we consider $\{\mathbf{P}^i\}$. When transmitter $i$ stays at the energy state $u\in\mathbf{\Delta}_i$, it can transfer to state $0$, $v_1$, or $v_2$ , where $v_1=\min\left\{u+E^iL,B_{max}\delta\right\}$, and $v_2=\min\left\{u+E^il,B_{max}\delta\right\}$. Accordingly, from $\mathbf{\Sigma}_{u_i=u}$, there are three possible cases:
\begin{enumerate}
  \item $\mathbf{\Sigma}_{u_i=u}\rightarrow\mathbf{\Sigma}_{u_i=0}$: For each state $\mathbf{u}\in\mathbf{\Sigma}_{u_i=u}$, there is only one possible route to $\mathbf{\Sigma}_{u_i=0}$ with probability $Q_ip_{tr}^i(u)$ such that transmitter $i$ transmits and goes into state $0$. In fact, such transition probability does not change for any $\mathbf{u}\in\mathbf{\Sigma}_{u_i=u}$. Thus, by taking all possible states into account, the transition probability can be computed by
\begin{align}\label{pu0sigma}
  p^i_{u,0}&=\mathbb{P}\left\{\mathbf{\Sigma}_{u_i=u}\rightarrow\mathbf{\Sigma}_{u_i=0}\mid\mathbf{\Sigma}_{u_i=u}\right\}\nonumber\\
&=\frac{Q_ip_{tr}^i(u)\mathbb{P}\left\{\mathbf{\Sigma}_{u_i=u}\right\}}{\mathbb{P}\left\{\mathbf{\Sigma}_{u_i=u}\right\}}=Q_ip_{tr}^i(u),
\end{align}
which is equal to (\ref{pu0}).
  \item $\mathbf{\Sigma}_{u_i=u}\rightarrow\mathbf{\Sigma}_{u_i=v_1}$: For each state $\mathbf{u}\in\mathbf{\Sigma}_{u_i=u}$, there are $I-1$ possible routes to $\mathbf{\Sigma}_{u_i=v_1}$. We pick the route caused by transmitter $j\neq i$, i.e., the $j$-th transmitter transmits. Suppose that at state $\mathbf{u}$, the transmitter $j$ is in the energy state $bE^jl\in\mathbf{\Delta}_j$. The probability of staying at $\mathbf{\Sigma}_{u_i=u,u_j=bE^jl}$ is given as $\pi_b^j\mathbb{P}\left\{\mathbf{\Sigma}_{u_i=u}\right\}$ by (\ref{bdelta}). Thus, the transition $\mathbf{\Sigma}_{u_i=u,u_j=bE^jl}\rightarrow\mathbf{\Sigma}_{u_i=v_1, u_j=0}$ describes the transition of transmitter $i$ from state $u$ to state $v_1$ caused by transmitter $j$ with energy level $u_j=bE^jl$. Similarly as in (\ref{pu0sigma}), the transition probability for this case is given by
\begin{align}
  &\mathbb{P}\left\{\mathbf{\Sigma}_{u_i=u,u_j=bE^jl}\rightarrow\mathbf{\Sigma}_{u_i=v_1,u_j=0}\mid\mathbf{\Sigma}_{u_i=u,u_j=bE^jl}\right\}\nonumber\\
=&\frac{Q_jp_{tr}^j(bE^jl)\mathbb{P}\left\{\mathbf{\Sigma}_{u_i=u,u_j=bE^jl}\right\}}{\mathbb{P}\left\{\mathbf{\Sigma}_{u_i=u,u_j=bE^jl}\right\}}\nonumber\\
=&Q_jp_{tr}^j(bE^jl).\nonumber
\end{align}
When we extend to other transmitters besides $i$, and consider all possible states for each transmitter, we obtain the probability of the one step transition $\mathbf{\Sigma}_{u_i=u}\rightarrow\mathbf{\Sigma}_{u_i=v_1}$ as
\begin{align}
  &\mathbb{P}\left\{\mathbf{\Sigma}_{u_i=u}\rightarrow\mathbf{\Sigma}_{u_i=v_1}\mid\mathbf{\Sigma}_{u_i=u}\right\}~~~~~~~~~~~~~~~~~~~~\nonumber\\
  =&\frac{\mathbb{P}\left\{\mathbf{\Sigma}_{u_i=u}\rightarrow\mathbf{\Sigma}_{u_i=v_1},~\mathbf{\Sigma}_{u_i=u}\right\}}{\mathbb{P}\left\{\mathbf{\Sigma}_{u_i=u}\right\}}\nonumber
\end{align}
\begin{align}
=&\frac{1}{\mathbb{P}\left\{\mathbf{\Sigma}_{u_i=u}\right\}}\sum_{j\neq i}\sum_{b=0}^{B_{max}}\left(\mathbb{P}\left\{\mathbf{\Sigma}_{u_i=u,u_j=bE^jl}\right\}\right.\nonumber\\
&~~~~~~~\cdot\left.\mathbb{P}\left\{\mathbf{\Sigma}_{u_i=u}\rightarrow\mathbf{\Sigma}_{u_i=v_1}\mid\mathbf{\Sigma}_{u_i=u,u_j=bE^jl}\right\}\right)\nonumber
\end{align}
\begin{align}
=&\frac{1}{\mathbb{P}\left\{\mathbf{\Sigma}_{u_i=u}\right\}}\sum_{j\neq i}\sum_{b=0}^{B_{max}}\left(\mathbb{P}\left\{\mathbf{\Sigma}_{u_i=u,u_j=bE^jl}\right\}\right.\nonumber\\
\cdot&\left.\mathbb{P}\left\{\mathbf{\Sigma}_{u_i=u,u_j=bE^jl}\rightarrow\mathbf{\Sigma}_{u_i=v_1,u_j=0}\mid\mathbf{\Sigma}_{u_i=u,u_j=bE^jl}\right\}\right)\nonumber\\
=&\frac{1}{\mathbb{P}\left\{\mathbf{\Sigma}_{u_i=u}\right\}}\sum_{j\neq i}\sum_{b=0}^{B_{max}}\pi_b^j\mathbb{P}\left\{\mathbf{\Sigma}_{u_i=u}\right\}Q_jp_{tr}^j(bE^jl)\nonumber\\
=&\sum_{j\neq i}\sum_{b=0}^{B_{max}}\pi_b^j Q_jp_{tr}^j(bE^jl).\label{puv1equi}
\end{align}
Thus, (\ref{puv1equi}) is equivalent to (\ref{puv1}).
  \item $\mathbf{\Sigma}_{u_i=u}\rightarrow\mathbf{\Sigma}_{u_i=v_2}$: The transition probability for this case can be obtained by taking the complement of (\ref{pu0sigma}) and (\ref{puv1equi}), which is equivalent to (\ref{puv2}).
\end{enumerate}
Therefore, we obtain all possible transitions for transmitter $i$ at time $t$, for which the corresponding transition probabilities can be computed as well. Thus, $\{\Pi^i\}$ and $\{\mathbf{P}^i\}$ are obtained from $\mathbf{\Pi}$ and $\mathbf{P}$, which proves the direction ``$\Rightarrow$'' of (\ref{equivalence}).

Overall, the convergence of Algorithm \ref{markoviter} is proved.

\end{document}